\documentclass[10pt,conference,a4paper]{IEEEtran}
\usepackage[left=1.3cm, right=1.30cm, top=2cm, bottom=1cm]{geometry}
\usepackage[fleqn]{amsmath}
\usepackage[pdftex]{graphicx}
\usepackage{amssymb}
\usepackage{amsthm}
\usepackage{bm}
\usepackage{color}
\usepackage{dsfont}
\usepackage{epstopdf}
\usepackage{verbatim}
\usepackage{textcomp}
\usepackage{tikz}
\newtheorem{thm}{Theorem}

\newtheorem{rem}{Remark}
\graphicspath{ {Figures/} }

\title{  Jointly Broadcasting Data and Power with Quality of Service Guarantees }

\author{Deekshith P K \IEEEauthorrefmark{1}, Trupthi Chougule\IEEEauthorrefmark{2}, Shreya Turmari\IEEEauthorrefmark{2}, Ramya Raju\IEEEauthorrefmark{2}, Rakshitha Ram\IEEEauthorrefmark{2} and Vinod Sharma\IEEEauthorrefmark{1} \\ {\IEEEauthorrefmark{1}ECE Dept., Indian Institute of Science, Bangalore, India. \IEEEauthorrefmark{2}PES Institute of Technology, Bangalore, India.}\\{\hspace{-.6cm} \IEEEauthorrefmark{1}Email: \{deeks, vinod\}@ece.iisc.ernet.in},\\ \IEEEauthorrefmark{2}Email:\{chouguletrupthi, shreya.turmari, ramyaraju93, rakshitharam.93\}@gmail.com.}

\IEEEoverridecommandlockouts
\begin{document}
\maketitle
\begin{abstract}
In this work, we consider a scenario wherein an energy harvesting wireless radio equipment sends information to multiple receivers alongside powering them. In addition to harvesting the incoming radio frequency (RF) energy, the receivers also harvest energy from {its environment (e.g., solar energy)}. This communication framework is captured by a fading Gaussian Broadcast Channel (GBC) with energy harvesting transmitter and receivers. In order to ensure {some quality of service (QoS)} in data reception among the receivers, we impose a \textit{minimum-rate} requirement on data transmission. For the setting in place, we characterize the fundamental limits in jointly transmitting information and power subject to a QoS guarantee, for three cardinal receiver structures namely, \textit{ideal}, \textit{time-switching} and \textit{power-splitting}. We show that a time-switching receiver can {switch} between {information reception mode} and {energy harvesting mode}, \textit{without} the transmitter's knowledge of the same and \textit{without} any extra \textit{rate loss}. We also prove that, for the same amount of power transferred, on average, a power-splitting receiver supports higher data rates compared to a time-switching receiver.  
\end{abstract}
\noindent
\section{Introduction}
Intentionally transferring energy along with information, using radio frequency (RF) signals, is an attractive alternative to perpetually and remotely power energy harvesting sensors  that have limited physical accessibility. It is foreseeable that in next generation wireless systems, a picocell or femtocell base station will be enabled to wirelessly charge low power communication devices within its range. These base stations themselves could be energy harvesting \textit{green} base stations. Apart from the numerous system design challenges the problem offers, it also opens up a rich set of theoretically motivated research avenues. {On} this premise, we address the problem of characterizing the fundamental limits in jointly broadcasting data and power over a wireless medium with energy harvesting transmitter and receivers.  
\par In this work, we consider the problem of Simultaneous Wireless Information and Power Transfer (SWIPT) over a fading Gaussian broadcast channel (GBC) with an energy harvesting transmitter, ensuring a certain quality of service (QoS) guarantee to the receivers. The QoS parameter we refer to is that of \textit{minimum-rate} constraint. For the canonical fading GBC (non energy harvesting), the problem of characterizing the fundamental limits with minimum rate constraints as a means to ensure \textit{fairness} among receivers is a well studied topic (\cite{jindal2003capacity}). In the context of SWIPT, the above constraint has the added advantage that the transmission ensures {a} \textit{minimum instantaneous RF power} at the receivers at all times (which can potentially be harvested).   
\par We provide an overview of the related literature so as to elucidate our contributions in its context. The idea of transmitting power using an information encapsulated data symbol goes back to \cite{varshney2008transporting}. In \cite{grover2010shannon}, the optimal trade-off between {the} achievable rate and {the} power transferred across a noisy coupled-inductor circuit is discussed. Capacity-energy regions of a discrete memoryless multiple access channel and a {multi-hop} channel with a single relay is characterized in \cite{fouladgar2012transfer}. Achievable rates over an \textit{uplink} channel wherein the transmitters are powered via RF signals in the \textit{downlink} are provided in \cite{hadzi2014multiple}. In \cite{amor2015feedback}, authors report a result on feedback enhancing the \textit{rate-energy region} over a \textit{constant gain} multiple access channel with simultaneous transmission of information and power. 
\par  In the context of MIMO systems, \cite{zhang2013mimo} studies SWIPT by a transmitter to two receivers (either \textit{spatially separated} or \textit{co located}) in which one receiver harvests energy and the other receiver decodes information. A  joint transmit beamforming and receive power-splitting design for downlink SWIPT system is studied in \cite{shi2014joint}. A joint information and power transfer scheme that encodes information in the receive antenna index and power transfer intensity is pursued  in \cite{zhang2015energy}. As for the receiver design for SWIPT systems, two practical receiver architectures proposed in literature are time-switching receivers and power-splitting receivers \cite{zhang2013mimo}. The power splitting can happen non adaptively or adaptively (\cite{liu2013wireless}, \cite{zhou2013wireless}). For a comprehensive survey of recent advances in the domain of RF energy harvesting networks, refer \cite{lu2015wireless}. A survey on SWIPT communication systems can be found in \cite{krikidis2014simultaneous}. 
\par Existing works, to the best of our knowledge, do not consider the information theoretic characterization of fundamental limits of SWIPT systems \textit{with energy harvesting transmitter}. In the fading GBC setting that we consider, we assume both the transmitter and the receivers can harvest from a perennial ambient source. The receivers treat the transmitter as an RF energy source to meet additional energy requirements, if any. Another novel aspect in the model we propose is the inclusion of minimum-rate constraints in characterizing the fundamental limits of SWIPT systems. 
\par This paper is organized as follows. In Section \ref{S_Prel}, we present the system model and notation. Section \ref{S_Main} is devoted to explain the main results of this work. We derive the minimum-rate capacity region of the SWIPT system under consideration, with ideal, time-switching and power-splitting receivers. Numerical results are provided in Section \ref{S_NR}. We conclude in Section \ref{S_Conc}. {Proofs are sketched in the Appendices.} 

\section{ System Model and Notation}

\label{S_Prel}

\begin{figure}[h]
\begin{center}
\includegraphics[scale=0.44]{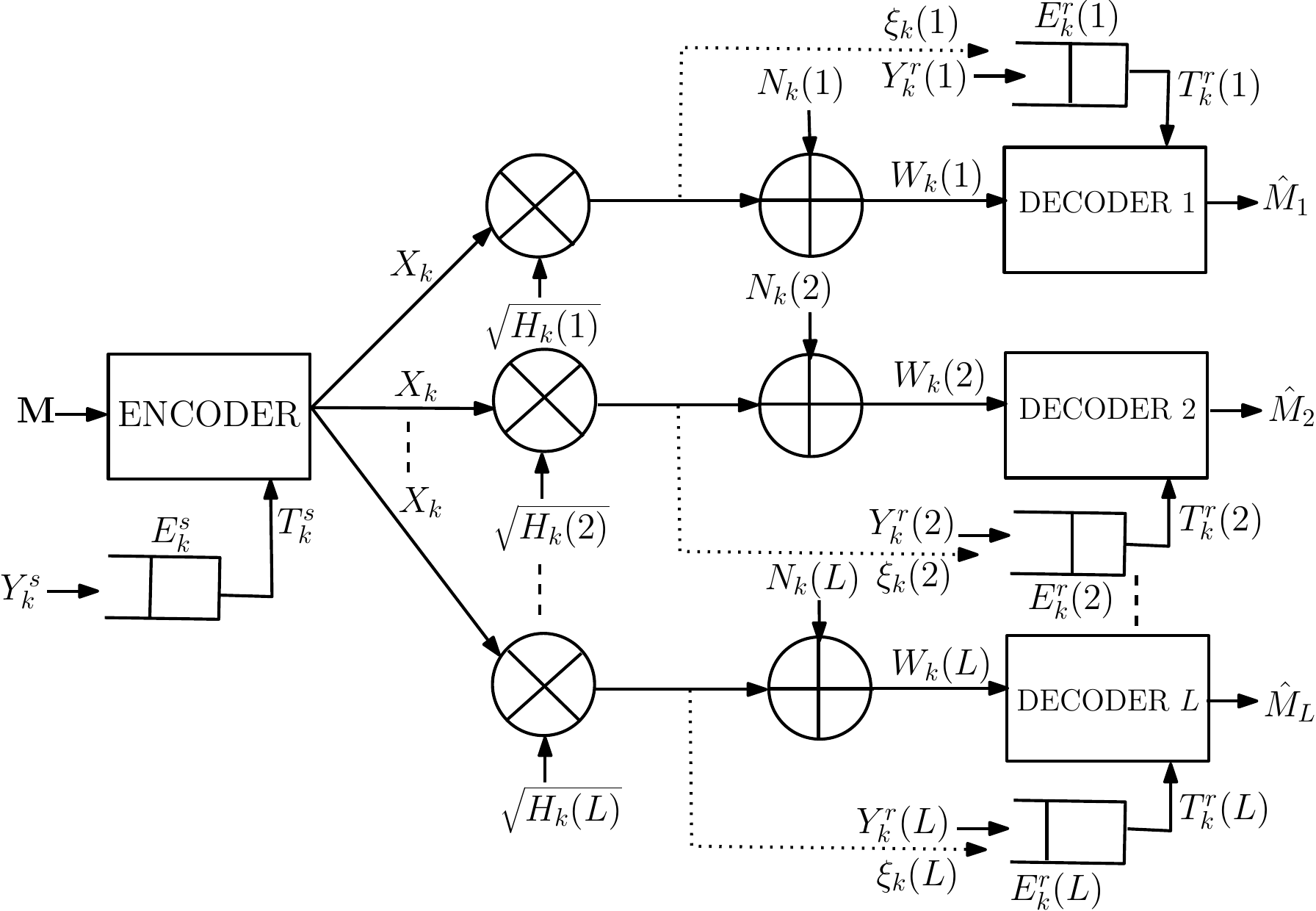}
\caption{A fading GBC with SWIPT.} \label{fig1}
\end{center}

\end{figure}

\subsection{Transmission with QoS Constraints}
Consider an energy harvesting transmitter equipped with an energy buffer (synonymous with battery or buffer) of infinite capacity. The transmitter could well be a \textit{green} base station harnessing renewable energy, like solar or wind energy. Alternatively, in the context of energy harvesting sensor networks, transmitter could represent a \textit{fusion centre} which multiple sensor nodes report to. 
\par We consider a time slotted system. In slot $k$, let $Y_k^s$ ($s$ indicates sender) denote the energy harvested by the transmitter from a renewable source. We assume the energy harvesting process $\{Y_k^s,~k\geq 1\}$ is stationary, ergodic. Let $\mathbb{E}[Y^s]$ denote the mean value of the energy harvesting process. Let $E_k^s$ denote the energy available in the transmitter's buffer at the beginning of slot $k$. In the model we consider, the harvested energy $Y_k^s$ can be used in the same slot and the remaining, if any, is stored in the buffer for future use. Let $T_k^s$ denote the energy used up by the transmitter in slot $k$. Thus, $T_k^s \leq \hat{E}_k^s \triangleq E_k^s+Y_k^s$. The energy in the buffer evolves according to $E_{k+1}^s=\hat{E}_{k}^s-T_k^s$.
\par The transmitter has $L$ messages to send, denoted by the message vector $\textbf{M} \triangleq (M_1,\hdots,M_L)$, to $L$ distant receivers, where $M_l$ is the message corresponding to receiver $l\in[1:L] \triangleq \{1,2,\hdots,L\}$. Simultaneously, the transmitter is powering each of the receivers. The receivers, in practice, could either be user mobile devices or low power sensor motes. Corresponding to the message vector $\mathbf{M}$,  a codeword of length $n$ $\big(X_1'(\mathbf{M}),\hdots,X_n'(\mathbf{M})\big) $,  is chosen. Since the transmitter is energy harvesting, transmitted symbol
in slot $k$ could be different from the codeword symbol {because of insufficient energy at the transmitter}.  The channel input symbol in slot $k$ is denoted as $X_k$. We note that the total energy used for transmission in slot $k$, $T_k^s=X_k^2=\sum\limits_{l=1}^LT_k^s(l)$, where $T_k^s(l)$ is the energy allocated for receiver $l$ in slot $k$. 
\par On account of the time varying nature of the underlying wireless channel, some users may be cut off from the {transmitter} for a certain duration of time depending upon the channel conditions. This is because the power allocation strategy which {ensures} the optimal \textit{long term} data rates will allocate zero transmission power in certain time slots to those users with low channel gains \cite{li2001capacity}. At the same time, it is not desirable to transmit at a target rate irrespective of the channel gains (essentially by a multi user variant of \textit{channel inversion}) as it reduces the permissible data rates \cite{li2001capacity2}. An alternative to the above approaches  is transmitting at a certain \textit{minimum instantaneous rate} irrespective of the channel conditions (there by ensuring certain fairness among receivers) and use the additional power to maximize the long term achievable data rates \cite{jindal2003capacity}.  Accordingly, let $\rho(l)$ be the minimum rate of transmission to be ensured {to receiver $l$}, irrespective of the channel conditions. The model parameter $\bm{\rho}\triangleq \big(\rho(1),\hdots,\rho(L)\big)$ dictates the quality of service guarantee on the joint data and power broadcast.
\subsection{The Channel Model }
\par The channel from the transmitter to the $l^{\text{th}}$ receiver is a fading channel corrupted by an independent and identically distributed (i.i.d.) additive Gaussian noise process $\{N_k(l),~k\}$ at the receiver. We denote the probability density function of $N_k(l)$ (having mean $0$ and variance $\sigma_l^2)$ by $\mathcal{N}(0,\sigma_l^2)$. The multiplicative channel gain from the transmitter to the  $l^{\text{th}}$ receiver in slot $k$ is denoted as $H_k(l)$. We assume that the fading process $\{\mathbf{H}_k,~k\geq 1\}$ is jointly stationary, ergodic, where $\mathbf{H}_k \triangleq \big(H_k(1),\hdots,H_k(L)\big) \in \mathcal{H}^L$ with stationary distribution $F_{\mathbf{H}}$. Here, $\mathcal{H}\subset \mathbb{R}^+$, the positive real axis and $\mathcal{H}^L$ is the Cartesian product $\mathcal{H}\times\hdots\times\mathcal{H}$ (L times).  In addition, we assume that the channel gains $H_k(l)$ and $H_k(j)$ are statistically independent for $l \neq j$ and are known to all the receivers and the transmitter {at time $k$}. We consider a block fading channel model wherein the channel gain from the transmitter to receiver $l$ remains fixed for the duration of a \textit{channel coherence time} $T_c(l)$. The codeword length $n$ is assumed to be an integer multiple of the least common multiple of $\{T_c(l)$, $l \in [1:L]\}$.  If $W_k(l)$ is the channel output at receiver $l$ in slot $k$, $W_k(l)=\sqrt{H_k(l)}X_k+N_k(l).$ 
\par Note that, ensuring a certain minimum non-zero transmission power to all users in all time slots (dictated by $\bm{\rho}$), potentially requires infinite power if the fading process, with {positive} probability, can take values arbitrarily close to zero. As an example, for the same reason,  the  zero outage capacity region for a Rayleigh fading GBC is \textit{null} (\cite{li2001capacity2}). To encompass those transmission schemes that require a finite average power and ensure non-zero minimum rate, we make the assumption that $\mathbb{E}[\frac{1}{H_k(l)}]<\infty$ for all $l$ and for all $k\geq 1$.      
\subsection{Receiver}
\par The receivers for SWIPT serve a dual purpose. There is a communication receiver to receive and decode the incoming data, and a rectenna module to harvest the RF energy. In slot $k$, receiver $l$ harvests $Y_k^r(l)$ ($r$ denotes receiver) from a surrounding perennial source. We assume that $\{Y_k^r(l),~k\geq 1\}$ is a stationary, ergodic process for each $l$ and is independent across receivers. Let $\mathbb{E}[Y^r(l)]$ denote its mean value. The receivers have an energy buffer of infinite capacity.  Let $E_k^r(l)$ denote the energy in $l^{\text{th}}$ receiver's buffer at the beginning of slot $k$. Let $\hat{E}_k^r(l)\triangleq E_k^r(l)+Y_k^r(l)$. There are various sources of energy consumption at the receivers. The front end of the communication receiver requires energy for filtering and other processing operations. This energy requirement, at receiver $l$, is modelled by a stationary, ergodic process $\{T_k^r(l),~k \geq 1\}$. We refer to {$\Delta_l \triangleq (\mathbb{E}[T^r(l)]-\mathbb{E}[Y^r(l)])^+$}  as the average energy deficit at receiver $l$, where $(.)^+=\max\{0,.\}$. Receiver $l$ harvests, on an average, $\Delta_l$ units of RF energy so as to compensate for the deficit. 
\par We now provide a brief description of the various receiver architectures considered in this work. An ideal receiver can harvest the incoming RF energy without \textit{distorting} the noise corrupted data symbol. Thus, the total energy harvested $D_k^r(l)$ at receiver $l$ in slot $k$ is $Y_k^r(l)+\xi_k(l)$, where $\xi_k(l)=\eta H_k(l)X_k^2$. Here, $\eta$ denotes the efficiency factor of {the} energy harvesting system (\cite{lu2015wireless}). The fundamental limits obtained in \cite{varshney2008transporting}, \cite{fouladgar2012transfer} are achievable only using ideal SWIPT receivers. In contrast, a time-switching receiver harvests RF energy in {a} slot, at the expense of erasing the corresponding noise corrupted data symbol. Let $\mathcal{I}_{l,k}$ denote the indicator function of the event that RF energy is harvested by receiver $l$ {in slot} $k$. Then, $D_k^r(l)=Y_k^r(l)+\mathcal{I}_{l,k}\xi_k(l)$. A power-splitting receiver \textit{divides} the incoming RF power between the communication module and the rectenna, non-adaptively. We refer to it as the constant fraction power-splitting receiver. If $0 \leq \pi_{\mathcal{E}} \leq 1$ is the fraction of power harvested in every slot, $D_k^r(l)=Y_k^r(l)+\pi_{\mathcal{E}}\xi_k(l)$.
\par In general, among the $L$ receivers, some receivers could be ideal and some others could be time-switching or power-splitting. But for the sake of exposition, we derive results assuming that all the receivers belong to one of the above kind. Our proof techniques readily yield the corresponding results for the general case as well.
\par In this work, we derive the fundamental limits in the framework propounded in \cite{rajesh2014capacity}. Specifically, the channel input and output processes  need not be stationary, since at time $k=0$, the transmitter and the receivers start operating with \textit{arbitrary} initial energy in their buffers. We consider \textit{power control policies} (to combat fading) at the transmitter such that the stochastic process $\{T_k^s,~k \geq 1\}$ is an  asymptotic mean stationary (AMS), ergodic process \cite{gray2011entropy}. We prove that, for the SWIPT system in place, the \textit{AMS capacity region} is equivalent to the \textit{Shannon capacity region} of a non energy harvesting system  with the same average power constraints.

\section{Minimum-Rate Capacity Region with Various Receiver Architectures} 
\label{S_Main}
In this section, we derive the minimum-rate capacity regions of the SWIPT system for the three receiver models. We begin {with} the following definitions. An energy management policy $T_{k}^s$ is called a {\textit{Markovian policy}}, if it is exclusively a function of the variables $\hat{E}_k^s$ and $\mathbf{H}_k$. In this work, we only consider policies that are Markovian. {We refer to such policies as Markovian because, if the processes $\{Y_k^s\}$, $\{\textbf{H}_k\}$, $\{Y_k^r(l),l\in[1:L]\}$, $\{T_k^r(l),l\in[1:L]\}$ are i.i.d., adopting Markov policies make the joint process $\{\big(Y_k^s,E_k^s,X_k^s,W_k(1),\hdots,W_k(L)\big)\}$ a Markov process}. We prove that such policies are \textit{optimal}  among the class of AMS, ergodic policies. A rate tuple $\mathbf{R}=\big(R(1),\hdots R(L)  \big)$ is {\textit{achievable}} if there {exists} a sequence of  $\big((2^{nR_1},\hdots, 2^{nR_L} ),n\big)$ codes, an encoding function,  a power controller so that for each joint fading state $\mathbf{h}=\big(h(1),\hdots h(L)\big)$, the instantaneous rate vector $\mathbf{R}(\mathbf{h}) \triangleq \big(R_1(\mathbf{h}),\hdots R_L(\mathbf{h})\big)$ satisfies $R_i(\mathbf{h})\geq \rho(i)$ and $\mathbb{E}_{\mathbf{H}}[R_i(\mathbf{H})]\geq R(i)$, $L$ decoders and energy harvesters, such that the average probability of decoding error (averaged over all possible realizations of codebooks) $P_{e}^{(n)} \rightarrow 0$ as $n \rightarrow \infty$. {\textit{Minimum-rate capacity region}} is the closure of the set of all achievable rate vectors.

We note that the minimum-rate vector $\bm{\rho}$ should be within the zero outage capacity region \cite{li2001capacity2} of a fading GBC with {the} peak power constraint corresponding to the minimum peak power imposed by the energy harvesting process. Since non zero minimum rates can be ensured only if the energy harvesting process $\{Y_k^s\}$ at the transmitter is such that $Y_k^s>\delta$ a.s. for some small $\delta>0$ for all $k$, we assume the same.
\subsection{ SWIPT System:  Ideal Receivers}
\label{SS_ID}
We now provide a characterization of the minimum-rate capacity region when all receivers are assumed to be ideal. Let $\Sigma_l(\mathbf{H})\triangleq H(l)T_l'^s(\mathbf{H})$, $\nu_l(\mathbf{H})\triangleq \sigma^2_l+\sum\limits_{j =1}^{L}H(j)T_j'^s(\mathbf{H})\mathds{1}_{\mathcal{E}_{l,j}}$, where $\mathds{1}_{\mathcal{E}_{l,j}}$ is the indicator function corresponding to the event $\mathcal{E}_{l,j} \triangleq \{\sigma_l^2H(j)>\sigma_j^2H(l)\}$, $T_l'^s$ is an energy allocation policy corresponding to receiver $l$ and $\text{SNR}_l(\mathbf{H})\triangleq \Sigma_l(\mathbf{H})/\nu_l(\mathbf{H})$. Define
\begin{flalign*}
\hspace{-20pt}
\mathcal{C}_i(\mathbf{T}'^s)  =  \Big\{\mathbf{R}:~ & {\rho}(l) \leq {R}(l) \leq  \mathbb{E}_\mathbf{H}\Big[\mathbf{C}_{i,l}(\mathbf{H})\Big],~ l \in [1:L] \Big\}.
\end{flalign*}
Here, $\mathbf{C}_{i,l}(\mathbf{H}) \triangleq \frac{1}{2}\log\big(1+\text{SNR}_l(\mathbf{H})\big)$ and  $\mathbf{T}'^s=(T_1'^s,\hdots T_L'^s)$. For $\bm{\Delta} \triangleq (\Delta_1,\hdots,\Delta_L)$, let $$\mathcal{T}^s(\bm{\Delta})  \triangleq  \Big\{\mathbf{T}'^s: \mathbb{E}_\mathbf{H}\big[\sum\limits_{l=1}^{L}T_l'^s(\mathbf{H})\big]\leq \mathbb{E}[Y^s],~R_l(\mathbf{H}) \stackrel{\text{a.s.}}{\geq} \rho(l),$$ $$\mathbb{E}_\mathbf{H}\big[\eta H(l)T_l'^s(\mathbf{H})\big] \geq \Delta_l,~ l \in [1:L]\Big\}.$$

\begin{thm}
\label{cap_region}
\emph{(Capacity Region with Ideal Receivers)}: The minimum rate capacity region is 
\begin{flalign*}
\hspace{35pt}\mathcal{C}_i(\bm{\Delta}) = \overline{\text{Conv}}\Bigg(\bigcup_{\mathbf{T}'^s \in \mathcal{T}^s(\bm{\Delta})} \mathcal{C}_i(\mathbf{T}'^s)\Bigg),
\end{flalign*}
where $\overline{\text{Conv}}(S)$ is the closure of convex hull of the set $S$.
\end{thm}
\begin{proof}
See Appendix A.
\end{proof}
Since the capacity region $\mathcal{C}_i(\bm{\Delta})$ is convex, we can obtain the boundary points of $\mathcal{C}_i(\bm{\Delta})$ by solving the following optimization problem:
\begin{flalign*}
\hspace{20pt}
&\max_{\mathbf{T}'^s(.)} ~ \sum\limits_{l=1}^L \mu(l) \mathbb{E}_{\mathbf{H}}\big[R_l\big(T_l'^s(\mathbf{H})\big)\big],\\
& \text{s.t.}~ \mathbb{E}_\mathbf{H}\big[\sum\limits_{l=1}^{L}T_l'^s(\mathbf{H})\big]\leq \mathbb{E}[Y^s], ~ \forall l, \\
& ~R_l(\mathbf{H}) \stackrel{\text{a.s.}}{\geq} \rho(l),~\forall ~ l,\\
& \mathbb{E}_\mathbf{H}\big[\eta H(l)T_l'^s(\mathbf{H})\big] \geq {\Delta_l},~\forall ~ l.
\end{flalign*}
Let $\Pi(.)$ be a permutation function on $[1:L]$ such that $H\big(\Pi(1)\big)/\sigma_{\Pi(1)}^2\geq H\big(\Pi(2)\big)/\sigma_{\Pi(2)}^2\geq \hdots\geq H\big(\Pi(L)\big)/\sigma_{\Pi(L)}^2.$ Also, let $T_{m,l}'$ ($m$ indicates minimum) denote the energy expended for maintaining the minimum rate $\rho(l)$ and $T_{e,l}'$ ($e$ denotes excess) be the excess energy such that $T_{m,l}'+T_{e,l}'=T_{l}'^s$. Then, with additional algebraic manipulation, it is easy to show that the above optimization is equivalent {to the optimization problem:}
\begin{flalign*}
\hspace{20pt}
&\max_{\mathbf{T}'_e} ~ \sum\limits_{l=1}^L \Big[\mu(l)\rho(l) + \mathbb{E}_{\mathbf{H}_{\text{ef}}}\big[C_{l}^{\text{ef}}(\mathbf{H}_{\text{ef}})\big]\Big],\\
& \text{s.t.}~ \mathbb{E}_{\mathbf{H}_{\text{ef}}}\big[\sum\limits_{l=1}^{L}T_{e,l}'(\mathbf{H}_{\text{ef}})\big]\leq \mathbb{E}[Y^s], ~ \forall l, \\
& {\mathbb{E}_{\mathbf{H}_{\text{ef}}}\big[\eta H_{\text{ef}}(l)T_{e,l}'(\mathbf{H}_{\text{ef}})\big] \geq {\Delta_{l,\text{ef}}},~\forall ~ l},
\end{flalign*}
where, $C_{l}^{\text{ef}}(\mathbf{H}_{\text{ef}}) \triangleq \frac{1}{2}\log\big(1+\text{SNR}_{l,\text{ef}}(\mathbf{H}_{\text{ef}})\big),$ $\text{SNR}_{l,\text{ef}}(\mathbf{H}_{\text{ef}})=\Sigma_{l,\text{ef}}(\mathbf{H}_{\text{ef}})/\nu_{l,\text{ef}}(\mathbf{H}_{\text{ef}})$,  $\Sigma_{l,\text{ef}}(\mathbf{H}_{\text{ef}})=H_{\text{ef}}\big(\Pi(l)\big)T_{e,\Pi(l)}'(\mathbf{H}_\text{ef})$ and $\nu_{l,\text{ef}}(\mathbf{H}_{\text{ef}})=\sigma_{l,\text{ef}}^2+\sum\limits_{j <l}H_{\text{ef}}\big(\Pi(j)\big)T_{e,\Pi(j)}'(\mathbf{H}_{\text{ef}})$. We refer to $\mathbf{H}_{\text{ef}}$, $(\sigma_{l,\text{ef}}^2,l \in [1:L])$ as the effective fading coefficients and noise variances respectively,  and can be obtained as in \cite{jindal2003capacity}. {Also, $\mathbf{T}'_e \triangleq (T_{e,1}',\hdots,T_{e,L}')$, $\Delta_{l,\text{ef}}=\Delta_l- \mathbb{E}_{\mathbf{H}}\big[\eta H(l)T_{m,l}'(\mathbf{H})\big]$. As an example, for the two receiver case, let $q$ denote the  probability of the event $\mathcal{E}_{1,2}$ and let, $q_c=(1-q)$. Denote, for $l \in \{1,2\}$, $p_l=(e^{2\rho(l)}-1)$. Then, under the event $\mathcal{E}_{1,2}$, $\sigma_{1,\text{ef}}^2=\sigma_1^2$, $\sigma_{2,\text{ef}}^2=(\sigma_2^2-\sigma_1^2)e^{-2\rho(1)}+\sigma_1^2$,  $H_{\text{ef}}(l)=H(l)e^{-2\rho(1)-2\rho(2)},$  $\Delta_{1,\text{ef}}=\Delta_1-\sigma_1^2p_1-\sigma_2^2p_1p_2q_c$, $\Delta_{2,\text{ef}}=\Delta_2-\sigma_2^2p_2-\sigma_1^2p_1p_2q$. For the complement of the event $\mathcal{E}_{1,2}$, the indices are swapped to the obtain corresponding expressions.}
\begin{rem}
\label{rem1}
 As a consequence of Theorem \ref{cap_region}, we can recover various important results as special cases. For instance, the capacity region of a fading GBC with an energy harvesting transmitter, and without power transfer and minimum rate constraints, is readily obtained. We  also obtain the capacity of a fading AWGN channel with energy harvesting transmitter, sending simultaneously a delay sensitive data (at a pre specified rate $\rho$) and a delay tolerant data. The result can be obtained using the proof of Theorem \ref{cap_region}, but using two separate codebooks (for each class of data) in conjunction with the rate splitting argument \cite{rimoldi1996rate}.   
 \end{rem}

\subsection{SWIPT System: Time-Switching Receivers}
\label{S_TS}
In this section, we  consider the SWIPT system with time-switching receivers.  The corresponding capacity region is referred to as the minimum-rate erasure capacity region. The terminology signifies the fact that harvesting energy from a data bearing symbol (using time-switching receiver) erases its information content. {In time-switching case, even though there is \textit{no minimum rate} at those times}, the constraint ensures that receivers can harvest a certain minimum RF power. An important aspect of our model is that, without loss of \textit{optimality}, each receiver can decide when to harvest RF energy independent of other receivers' decision and the transmitter not knowing the same.  The probability with which receiver $l$ decides to harvest in any slot is dictated by $\Delta_l$.  Let $\pi_{\mathcal{E}}(l) \triangleq \Delta_l/ \mathbb{E}_\mathbf{H}\big[\eta H(l)T_l'^s(\mathbf{H})\big]$  and denote $\pi^c_{\mathcal{E}}(l)=1-\pi_{\mathcal{E}}(l)$. Let
\begin{flalign*}
\hspace{-20pt}
\mathcal{C}_t^e(\mathbf{T}'^s)  \triangleq  \Big\{\mathbf{R}:~ & {\rho}(l) \leq {R}(l) \leq  \mathbb{E}_\mathbf{H}\Big[\mathbf{C}_{t,l}(\mathbf{H})\Big],~ l \in [1:L] \Big\},
\end{flalign*}
where, $\mathbf{C}_{t,l}(\mathbf{H}) \triangleq \frac{\pi_{\mathcal{E}}^c(l)}{2}\log\big(1+\text{SNR}_l(\mathbf{H})\big)$.

\begin{thm}
\label{Th_TSR}
\emph{(Capacity Region with Time-Switching Receivers)}: 
\begin{flalign*}
\hspace{35pt}\mathcal{C}_t^e(\bm{\Delta}) = \overline{\text{Conv}}\Bigg(\bigcup_{\mathbf{T}'^s \in \mathcal{T}^s(\bm{\Delta})} \mathcal{C}_t^e(\mathbf{T}'^s)\Bigg),
\end{flalign*}
is the minimum-rate erasure capacity region.
\end{thm}

\begin{proof}
See Appendix B.
\end{proof}

\subsection{SWIPT System:  Power-Splitting Receiver}
\label{S_PS}
At receiver $l$, let $\pi_{\mathcal{E}}(l)$ fraction of energy be harvested  in every slot, where $\pi_{\mathcal{E}}(l)$ is defined as in the time-switching case. Let $\tilde{\nu}_l(\mathbf{H})=\sigma^2_l+\sum\limits_{j =1}^{L}\pi_{\mathcal{E}}^c(j)H(j)T_j'^s(\mathbf{H})\mathds{1}_{\tilde{\mathcal{E}}_{l,j}}$, where $\mathds{1}_{\tilde{\mathcal{E}}_{l,j}}$ is the indicator function corresponding to the event $\{\sigma_l^2\pi_{\mathcal{E}}^c(j)H(j)>\sigma_j^2\pi_{\mathcal{E}}^c(l)H(l)\}$. Also, let $ \overset{\sim}{\text{SNR}}_l(\mathbf{H}) =\Sigma_l(\mathbf{H})/\tilde{\nu}_l(\mathbf{H})$.  Define
\begin{flalign*}
\hspace{-20pt}
\mathcal{C}_p(\mathbf{T}'^s)  \triangleq  \Big\{\mathbf{R}:~ & {\rho}(l) \leq {R}(l) \leq  \mathbb{E}_\mathbf{H}\Big[\mathbf{C}_{p,l}(\mathbf{H})\Big],~ l \in [1:L] \Big\},
\end{flalign*}
where, $\mathbf{C}_{p,l}(\mathbf{H}) \triangleq \frac{1}{2}\log\big(1+\pi_{\mathcal{E}}^c(l)\overset{\sim}{\text{SNR}}_l(\mathbf{H})\big)$.
\begin{thm}
\label{Th_PSR1}
\emph{(Capacity Region with Power-Splitting Receivers)}: The closure of
\begin{flalign*}
\hspace{35pt}\mathcal{C}_p(\bm{\Delta}) = \overline{\text{Conv}}\Bigg(\bigcup_{\mathbf{T}'^s \in \mathcal{T}^s(\bm{\Delta})} \mathcal{C}_p(\mathbf{T}'^s)\Bigg),
\end{flalign*}
 is the minimum-rate capacity region with constant fraction power-splitting receivers.
\end{thm}

\begin{proof}
The proof follows from the proof of Theorem \ref{cap_region} with the channel gain from the transmitter to $l^{\text{th}}$ receiver scaled by a factor $\pi_{\mathcal{E}}^c(l)$.
\end{proof}
 The boundary points of $\mathcal{C}_t^e(\bm{\Delta})$ and $\mathcal{C}_p(\bm{\Delta})$ can be obtained by solving optimization problems similar to that for the ideal case. 
 \begin{rem}
 In the absence of energy harvesting constraints, it is well known that a GBC and a Gaussian multiple access channel (GMAC) are \textit{duals} of each other \cite{jindal2004duality}. A similar result can be proved for these channels when powered by energy harvesting sources. On account of space constraints, we choose to avoid the technical details. Rather, we provide a numerical example in Section \ref{S_NR}. 
 \end{rem}

\section{Numerical Results}
\label{S_NR}
{In this section, we provide numerical examples to compare the minimum-capacity region of the SWIPT system  for the ideal, time-switching (TS) and power-splitting (PS) receivers. The time slot is considered in multiples of $1 \mu$sec. We consider a 2 user GBC with $\sigma_1^2=0.8$, $\sigma_2^2=1.6$. We assume i.i.d. fading, independent across users. The fading distribution at each user is chosen such that $\mathbb{E}[H(1)]=0.8$, $\mathbb{E}[H(2)]=0.5$. We consider a \textit{discretized} Rayleigh fading channel, obtained as follows: Fix an appropriate subset of the positive real axis. We choose the interval $[0,10]$. We discretize this set in steps of $.1$. For the channel between the transmitter and receiver $l$, the probability $p_{l}(h)$, $h \in \{.1,.2,\hdots,9.9\}$, is chosen such that $p_{l}(h)=Pr\big(H(l) \in [h-.1,h]\big)$ and $p_l(10)=Pr\big(H(l) \geq 9.9\big)$, where $H(l)$ is exponentially distributed. We take $\mathbb{E}[Y^s]=10W$. Fix $\rho(1)=300$ Kbps,  $\rho(2)=150$ Kbps. Let $\mathbb{E}[T^r(1)]=90\mu$W, $\mathbb{E}[T^r(2)]=50\mu$W and the energy deficits  $\Delta_1=60\mu\text{W} \approx -12$dBm, $\Delta_2=30\mu\text{W}\approx -15 $dBm. The efficiency factor $\eta$ is fixed to $10^{-4}$. }
\par{ In Figure \ref{all_rate_reg}, we compare the minimum-rate capacity regions of the SWIPT system with all receivers ideal, all TS and all PS, against the capacity region of \textit{all ideal receivers} case without minimum rate constraints and achievable rates without RF power transfer. For the example we consider, RF power transfer with ideal receivers enhances the data rates by 200$\%$ for receiver 1 and about $67\%$ for receiver 2, with respect to that with no RF power transfer. Also, due to concavity of the achievable rates as a function of the power expended, $\mathcal{C}_t^e(\bm{\Delta}) \subseteq \mathcal{C}_p(\bm{\Delta})$. 

\begin{figure}[h]
\begin{center}
\includegraphics[scale=0.33]{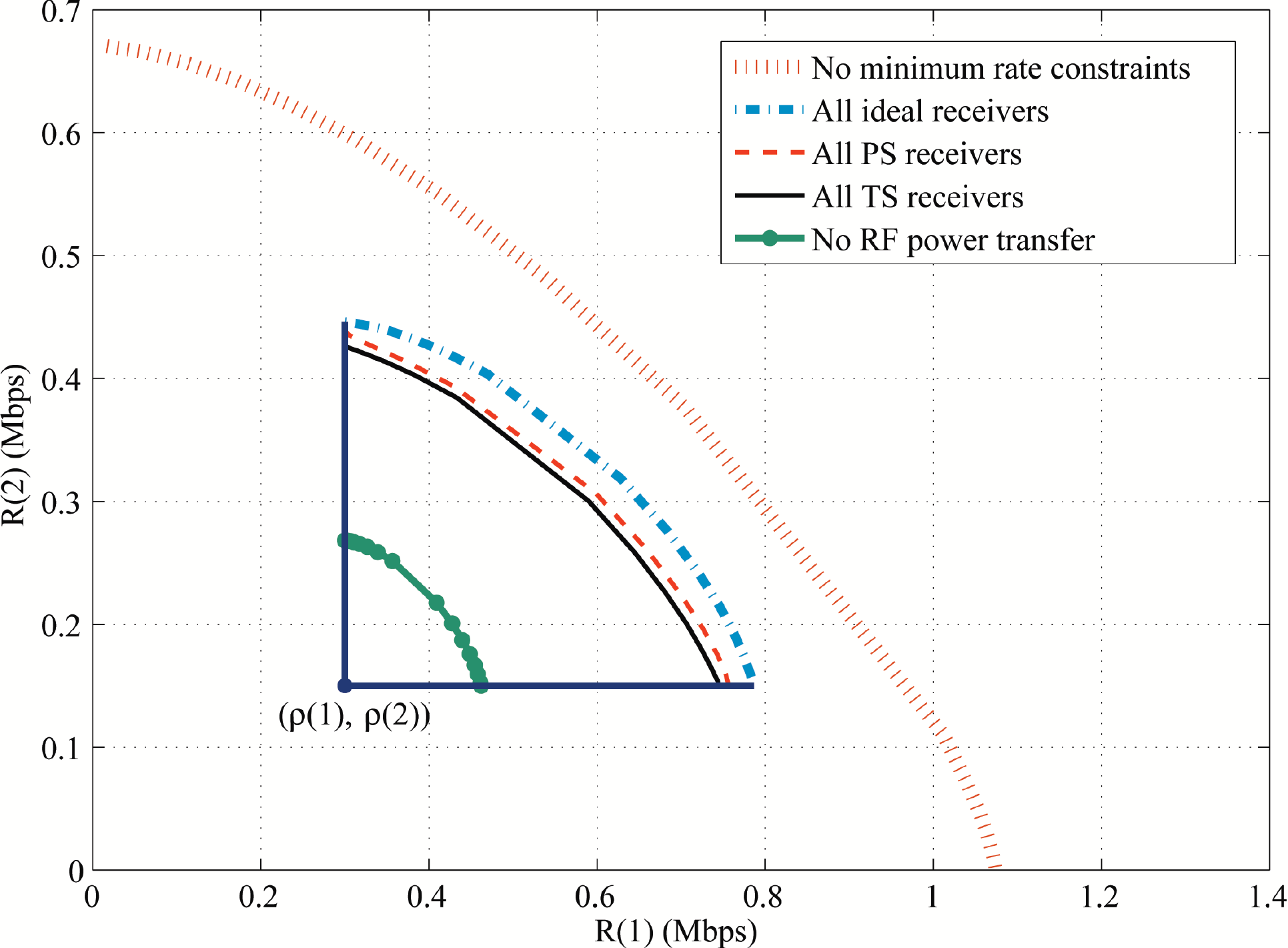}
\caption{ $\mathcal{C}_i(\bm{\Delta})$, $\mathcal{C}_t^e(\bm{\Delta})$ and $\mathcal{C}_p(\bm{\Delta})$ versus  capacity region without minimum rate constraints, achievable rates without RF power transfer.} \label{all_rate_reg}
\end{center}
\end{figure}

\par { We also compute the capacity region for a more realistic scenario in which receiver 1 is PS and receiver 2 is TS. We compare it with $\mathcal{C}_t^e(\bm{\Delta})$ and $\mathcal{C}_p(\bm{\Delta})$, in Figure \ref{PS_TS_PSTS_Cap_reg}, for two different values of $\mathbb{E}[Y^s]$. For the same amount of energy harvested at the transmitter, on average, a relatively wide range of energy deficit values at the receivers can be catered without \textit{much} rate loss. }
 \begin{figure}[h]
\begin{center}
\includegraphics[scale=0.33]{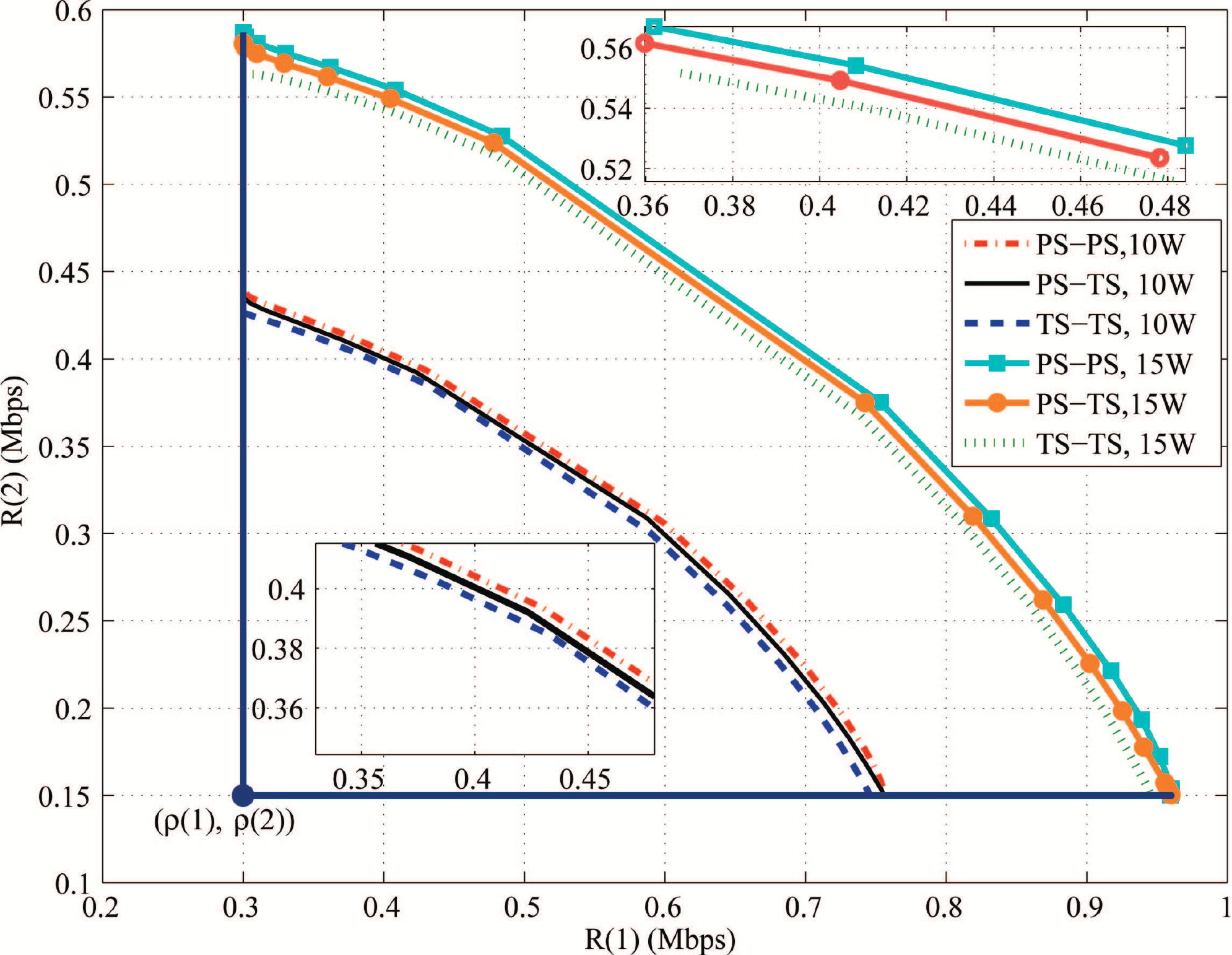} 
\caption{Comparison of capacity regions with receiver architectures all same and all different, for  $\mathbb{E}[Y^s]=10W$, $\mathbb{E}[Y^s]=15W$. } \label{PS_TS_PSTS_Cap_reg}
\end{center}
\end{figure}   
{Next, we fix $\mathbb{E}[Y^s]$ value and compare the data rates achievable for various values of energy deficits at the receiver. In Figure \ref{rr_TS_D}, we exemplify the change in $\mathcal{C}_t^e(\bm{\Delta})$ as a function of $\bm{\Delta}$. A similar plot for $\mathcal{C}_p(\bm{\Delta})$ is provided in Figure \ref{rr_PS_D}.} 

\begin{figure}[h]
\begin{center}
\includegraphics[scale=0.33]{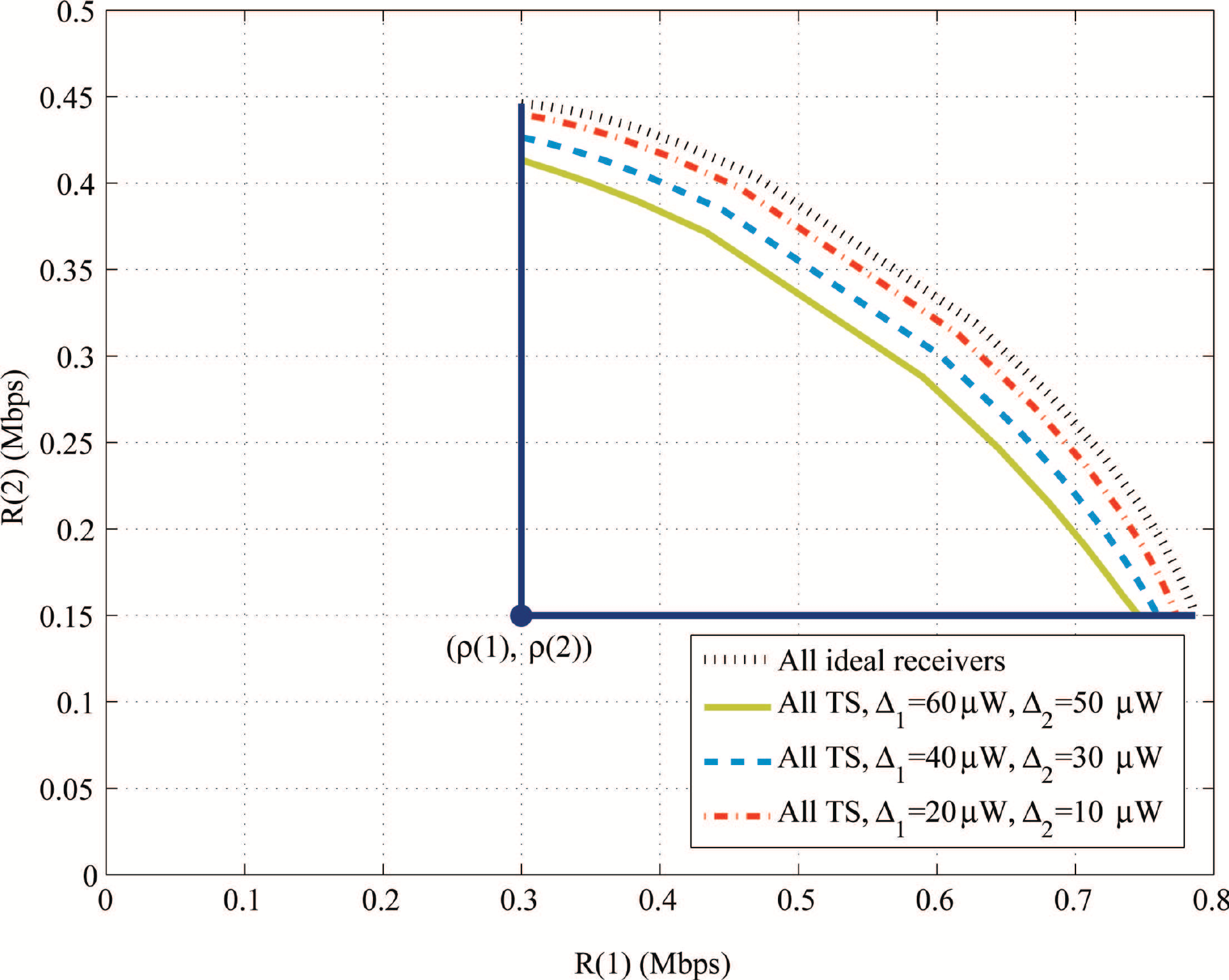}
\caption{Comparison of capacity regions for various values of $\bm{\Delta}$ with all TS receivers. } \label{rr_TS_D}
\end{center}
\end{figure} 

\begin{figure}[h]
\begin{center}
\includegraphics[scale=0.33]{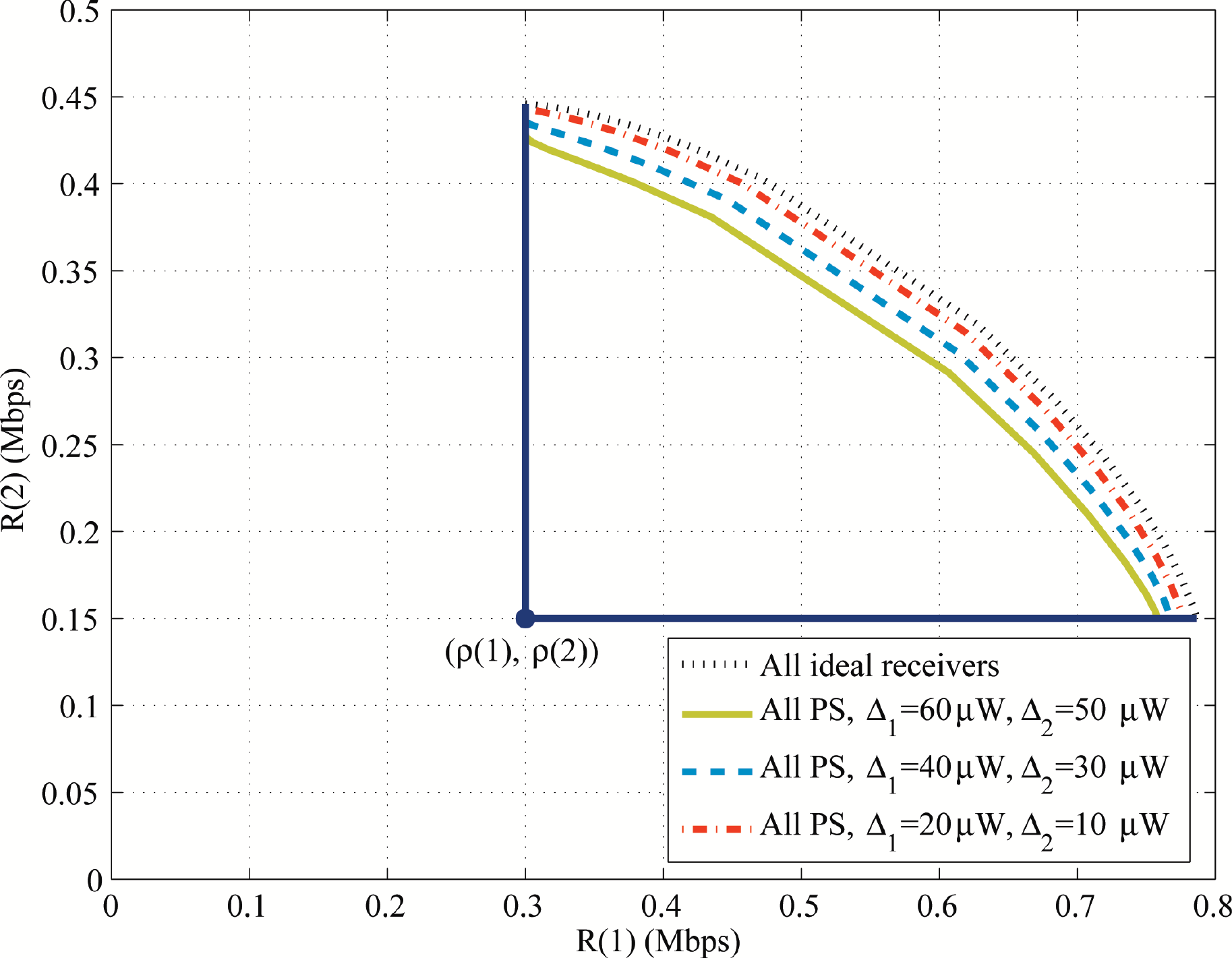}
\caption{Comparison of capacity regions for various values of $\bm{\Delta}$ with all PS receivers. } \label{rr_PS_D}
\end{center}
\end{figure}

\par {Finally, in Figure \ref{duality_fig}, we obtain the minimum-rate capacity region (with minimum rates as before) of a 2 user fading GMAC, with average energy harvested at the transmitters $\mathbb{E}[Y^s(1)]=6$W, $\mathbb{E}[Y^s(2)]=4$W, average energy consumed at the receiver (assumed to be ideal) $\mathbb{E}[T^r]=90\mu$W, energy deficit at the receiver $\Delta=60\mu$W, receiver noise variance $\sigma^2=1$, from that of a GBC with $\mathbb{E}[Y^s]=10$W, $\mathbb{E}[T^r(l)]=90\mu$W, $\Delta_l=60\mu$W and $\sigma_l^2=\sigma^2$, $l=1,2$. Even though the capacity regions are readily obtained  via duality, the method does not bring out the structure of the corresponding optimal power control policies.  
\begin{figure}[h]
\begin{center}
\includegraphics[scale=0.33]{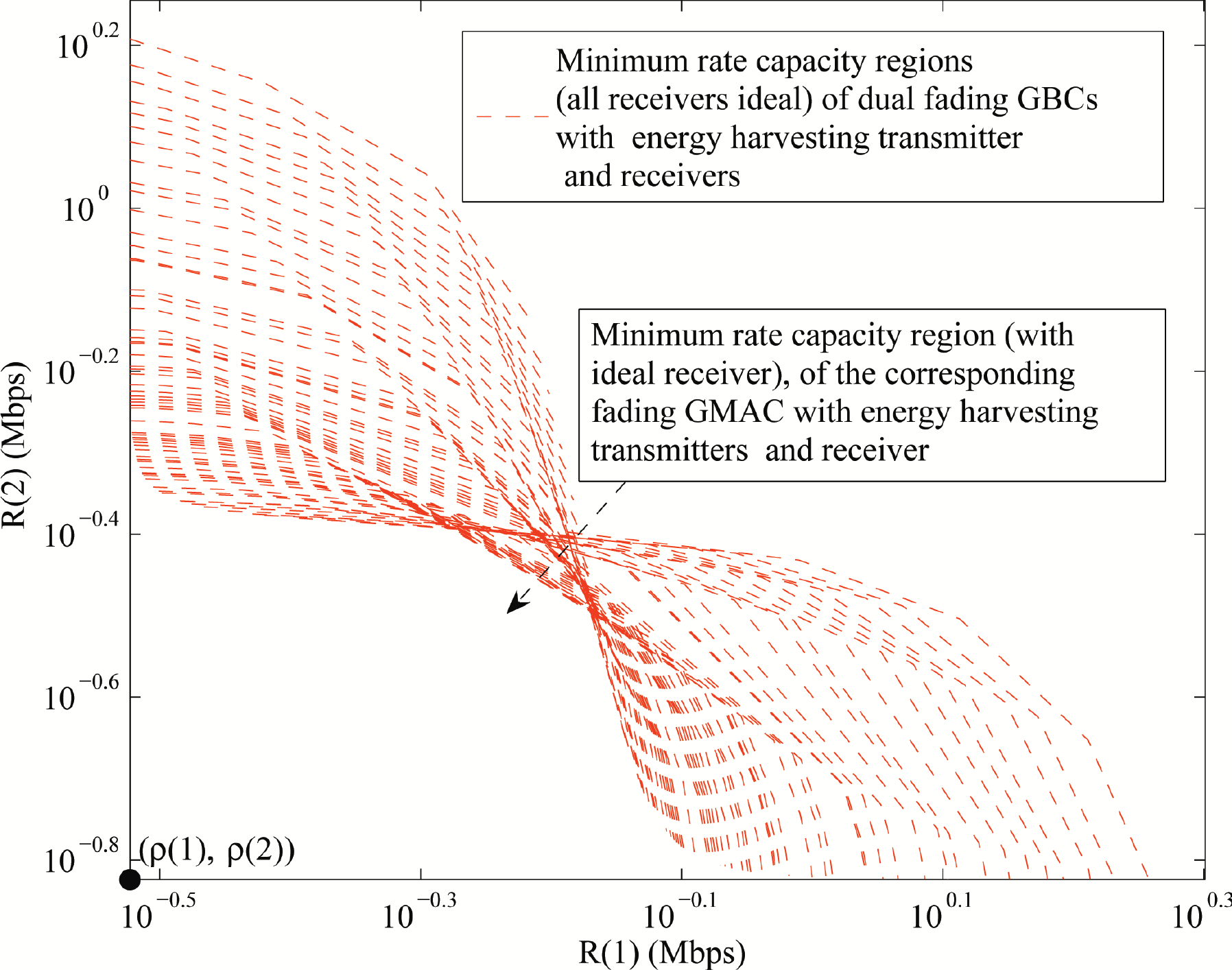}
\caption{The minimum-rate capacity region (in log scale) of  a fading GMAC with energy harvetsing constrainsts and SWIPT, via duality.} \label{duality_fig}
\end{center}
\end{figure}
\section{Conclusion}
\label{S_Conc}
In this work, we considered a fading GBC with an energy harvesting transmitter and receivers. We characterized the minimum-rate capacity region of the channel  with SWIPT for the ideal, time-switching and power splitting receiver architectures. The resultant power control policies obtained are optimal within a general class of permissible policies for energy harvesting SWIPT systems. Also, from our results in this work, we obtained numerically the corresponding capacity regions of a fading GMAC using duality arguments.}  
\section*{Appendix A}
\textit{Proof sketch of Theorem \ref{cap_region}:} {For the sake of clarity and brevity, we explain the proof technique for fading processes with finite support set $\mathcal{H}$. The extension to the continuous fading distributions can be handled in a standard way as in \cite{jindal2003capacity}.} \\ 
{Achievability:} \textit{Codebook Generation:} 
{Fix the power control policy $\mathbf{T}'^s$ obtained by solving the optimization problem in Section \ref{SS_ID} (with $\mathbb{E}[Y^s]$ therein, replaced by $\mathbb{E}[Y^s]-\epsilon$, for some small $\epsilon>0$). Fix message vector $\mathbf{M}$, blocklength $n$ and a rate vector $\mathbf{R}$. The message vector is divided into independent messages $\mathbf{M}_{\mathbf{h}}$ with rate $\mathbf{R}(\mathbf{h})$ such that $R(l)=\sum_{\mathbf{h}}R_l(\mathbf{h})$, $\mathbf{h}\in\mathcal{H}^L$. Corresponding to each joint fading state $\mathbf{h}$, there exists a unique order in which the channel is \textit{degraded}. That is, the receivers can be ordered according to the increasing values of $h(l)/ \sigma_l^2,$ $l\in[1:L]$ such that the receiver with the lowest value of $h(l)/ \sigma_l^2$ is the \textit{weakest receiver} and that with the highest value is the \textit{strongest}.  Accordingly, for each joint fading state (and the corresponding order of degradation), generate an $L$ level superposition codebook as per the \textit{satellization process} (\cite{bergmans1973random}, Section III B). Each of the $2^{nR_l(\mathbf{h})}$ codewords of the $l^{\text{th}}$ {satellite codebook} are generated i.i.d. according to $\mathcal{N}\big(0,T_l'^s(\mathbf{h})\big)$ and independent of other codebooks. The superposition codebooks generated are shared with all the receivers.} \\
\textit{Encoding and Signalling Scheme:} 
 {At time $k$, if the joint fading state is $\mathbf{h}_k$, the next untransmitted symbol in the codewords (to each of the receivers) corresponding to message $\mathbf{M}_{\mathbf{h}}$ is chosen for transmission. Since the transmitter is energy harvesting, in a given slot $k$, it may not have the required amount of energy $T'^s(\mathbf{h}_k)=\sum\limits_{l=1}^{L}T_l'^s(\mathbf{h}_k)$ in the buffer. In that case, transmission is done according to the following \textit{truncated policy}:}
 \begin{displaymath}
   T_k^s(l) = \left\{
     \begin{array}{lr}
      T_l'^s(\mathbf{h}_k) & :  T_l'^s(\mathbf{h}_k) \leq \hat{E}_k^s,\\
       \frac{\hat{E}_k^s}{T'^s(\mathbf{h}_k)}T_l'^s(\mathbf{h}_k) & : T_l'^s(\mathbf{h}_k) > \hat{E}_k^s.
     \end{array}
   \right.
\end{displaymath}
Since the average power expended at the transmitter  $\mathbb{E}[Y^s]-\epsilon$ is strictly less than the average harvested energy $\mathbb{E}[Y^s]$, $E_k^s \rightarrow \infty$ a.s. as $k \rightarrow \infty$ (Chapter 7, \cite{walrand1988introduction}). Accordingly, $T_k^s(l) \rightarrow T_l'^s $ a.s. as $k \rightarrow \infty$ for each $l$. 
\subsubsection*{Decoding}
{Since the channel gains are known perfectly, receiver $l$ can \textit{demultiplex} its received sequence $w^n(l)$ into subsequences $\{w^{n_{\mathbf{h}}}(l)\}$ such that $n=\sum_{\mathbf{h}}n_{\mathbf{h}}$. Note that, by the law of large numbers, $(n_{\mathbf{h}}/n) \geq (1-\delta)p(\mathbf{h})$, for a large $n$ and small $\delta>0$, where $p(\mathbf{h})$ is the probability of the joint fading state $\mathbf{h}$. Hence, for the demultiplexed subsequence corresponding to state $\mathbf{h}$ at each receiver, the decoding operation can be performed using a sub codebook of block length $n(1-\delta)p(\mathbf{h})$.  Each receiver adopts successive cancellation decoding. Note that, each $\mathbf{h}$ corresponds to a particular channel degradation order. Successive cancellation decoding corresponding to the degradation order of state $\mathbf{h}$ is performed such that, each receiver decodes  all the codewords (corresponding to message $\mathbf{M}_{\mathbf{h}})$ of all the receivers {degraded with respect it}, \textit{subtracts them off} and decodes its own codeword.}
\subsubsection*{Analysis of Error Events}
First note that, by ensuring $\eta\mathbb{E}[H(l)T'^s(l)] \geq \Delta_l$, the total mean harvested energy by an ideal receiver $l$, $\mathbb{E}[D^r(l)] \geq \mathbb{E}[T^r(l)]$. Thus, the probability of energy outages can be proven to \textit{vanish} asymptotically as in the Transmitter's case. Next, note  that AMS, ergodic sequences satisfy Asymptotic Equipartition Property (AEP) (\cite{barron1985strong}) under appropriate regularity conditions. These conditions hold good for the setting under consideration. In particular, the channel input and output random variables have finite variances. Also, the non energy harvesting channel with average transmitter power constraint equal to $\mathbb{E}[Y^s]$ has finitely bounded capacity region  and is an upper bound to the capacity of the system model under consideration. Hence, the associated mutual information rates in our case are  all finite. In addition, the AMS stationary mean is dominated by an i.i.d. Gaussian measure on a suitable Euclidean space. Thus the AEP result in (\cite{barron1985strong}) can be invoked in our context. Decoding is done with respect to the joint  finite dimensional distribution induced  by the stationary AMS mean distribution on the channel input and output processes. Taking into consideration these facts, error event analysis can be performed as in the standard case to obtain the required result.
\subsection*{Converse:} To prove the converse part, assume that there exist codebooks (satisfying the required constraints; in particular the minimum rate constraint is equivalent to a time varying constraint on the minimum transmitted power), encoder and decoders such that $P^{(n)}_e$ (average probability of decoding error)  goes to zero as $n \rightarrow \infty$.  For the energy harvesting transmitter, $\frac{1}{n} \sum\limits_{k=1}^{n}T_k^s \leq \frac{1}{n} \sum\limits_{k=1}^{n}Y_k^s\leq \mathbb{E}[Y^s]+\delta_{n_0} $ for $n>n_{n_0} $ and arbitrarily chosen small $\delta_{n_0}>0$. Hence, rates obtained via any coding, decoding scheme subject to the average power constraint $\mathbb{E}[Y^s]$ alone is an upper bound to the system capacity. This proves the converse. \qed

\section*{Appendix B}
\textit{Proof of Theorem \ref{Th_TSR}:}
At receiver $l$ fix an appropriate $\pi_{\mathcal{E}}(l) \in [0,1]$.  In each time slot,  receiver $l$ harvests RF energy with probability $\pi_{\mathcal{E}}(l)$. If energy is harvested, channel output is recorded as an erasure. Thus, the system with time-switching receiver can be equivalently thought of as a  fading GBC concatenated with an erasure channel. Encoding is done as per in the proof of Theorem \ref{cap_region}. Decoder discards the erasures and perform successive cancellation on the remaining. The erasures are independent of the channel output and the fraction of erasure instances converge almost surely to $\pi_{\mathcal{E}}(l)$.   Hence, the achievability follows as in the case of Theorem \ref{cap_region}.
\par To prove the converse, we can assume that the encoder has access to non-causal knowledge of erasure locations. The encoders can choose to send a zero symbol during the erasure instances. The decoder discards the erased channel outputs. This, along with the converse argument in Theorem \ref{cap_region} proves the converse for the time-switching receiver case. \qed

\bibliographystyle{IEEEtran}
\bibliography{bibfile_fb_capacity}
 \end{document}